\documentclass[11pt]{article}

\usepackage{fullpage}
\usepackage{graphicx,enumerate}
\usepackage{float}
\usepackage{natbib}

\usepackage{times}
\usepackage{soul}
\usepackage{url}
\usepackage[hidelinks]{hyperref}
\usepackage[utf8]{inputenc}
\usepackage[small]{caption}
\usepackage{graphicx}
\usepackage{amsmath}
\usepackage{amsthm}
\usepackage{booktabs}
\usepackage{algorithm}
\usepackage{libertine}
\usepackage{hyperref}
\usepackage[svgnames]{xcolor}
\usepackage[capitalise,nameinlink]{cleveref}
\hypersetup{colorlinks={true},linkcolor={DarkBlue},citecolor=[named]{DarkGreen}}
\allowdisplaybreaks

\usepackage{amsfonts}
\usepackage{amsthm}
\usepackage{enumitem}
\usepackage{bm}
\usepackage{xcolor}
\usepackage{todonotes}
\usepackage{float}
\usepackage{algpseudocode} 

\usepackage{lineno}

\newtheorem{theorem}{Theorem}
\newtheorem{definition}{Definition}

\newtheorem{lemma}{Lemma}

\newtheorem*{opt-cash-inj}{Optimal Cash Injection}
\newtheorem*{opt-debt-rem}{Optimal Debt Removal}

\newcommand{\greedy}{{\sc{Greedy}}}
\newcommand{\liq}{{\mathcal{F}}}

\newcommand{\fn}{{G}}
\newcommand{\poa}{\text{\normalfont PoA}}
\newcommand{\pos}{\text{\normalfont PoS}}
\newcommand{\eoa}{\text{\normalfont EoA}}
\newcommand{\eos}{\text{\normalfont EoS}}

\newcommand{\remove}[1]{}
\newcommand{\odr}{\textsc{Opt-Debt-Removal}}
\newcommand{\oci}{\textsc{Opt-Cash-Injection}}

\newcommand{\NP}{\textbf{NP}}

\newcommand{\pk}[1]{{\color{black}{#1}}}
\newcommand{\mk}[1]{{\color{black}{#1}}}
\newcommand{\hz}[1]{{\color{black}{#1}}}
\newcommand{\mph}[1]{{\color{black}{#1}}}

\begin{document}

\title{\bf Optimal Bailouts and Strategic Debt Forgiveness in Financial Networks\thanks{A preliminary version of this paper appeared in ~\citep{kanellopoulos2022forgiving}. The research was carried out when Hao Zhou was a PhD student at the University of Essex.}}




\author{Panagiotis Kanellopoulos\textsuperscript{1} \and
Maria Kyropoulou\textsuperscript{1} \and
Hao Zhou\textsuperscript{1,2}}

\date{
\textsuperscript{1}School of Computer Science and Electronic Engineering, University of Essex, UK\\
\textsuperscript{2}Department of Computer Science, University of Oxford, UK}

\maketitle


\begin{abstract}
A financial system is represented by a network, where nodes correspond to banks, and directed labeled edges correspond to debt contracts between banks. Once a payment schedule has been defined, 
the liquidity of the system is defined as the sum of total payments made in the network. Maximizing systemic liquidity is a natural objective of any financial authority, so, we study the setting where the financial authority offers bailout money to some bank(s) or forgives the debts of others in order to help them avoid costs related to default, and, hence, maximize liquidity.
We investigate the approximation ratio provided by the greedy bailout policy compared to the optimal one, and we study the computational hardness of finding the optimal debt-removal and budget-constrained optimal bailout policy, respectively.

We also study financial systems from a game-theoretic standpoint. We observe that the removal of some incoming debt might be in the best interest of a bank, if that helps one of its borrowers remain solvent and avoid costs related to default. 
Assuming that a bank's well-being (i.e., utility) is aligned with the incoming payments they receive from the network, we define and analyze a game among banks who want to maximize their utility by strategically giving up some incoming payments. In addition, we extend the previous game by considering bailout payments. After formally defining the above games, we prove results about the existence and quality of pure Nash equilibria, as well as the computational complexity of finding such equilibria.
\end{abstract}

\section{Introduction}\label{sec:introduction}
A financial system comprises a set of institutions, such as banks, that engage in financial transactions. The interconnections showing the liabilities (financial obligations or debts) among the banks can be represented by a network, where the nodes correspond to banks and the edges correspond to liability relations. Each bank has a fixed amount of external assets (not affected by the network) which are measured in the same currency as the liabilities. A bank's total assets comprise its external assets and its incoming payments, and may be used for (outgoing) payments to its lenders. If a bank's assets are not enough to cover its liabilities, that bank will be in default and the value of its assets will be decreased (e.g., by liquidation); the extent of this decrease is captured by \emph{default costs} and essentially implies that the corresponding bank will have only a part of its total assets available for making payments. 

On the liquidation day (also known as clearing), each bank in the system has to pay its debts in accordance with the following three principles of bankruptcy law (see, e.g., \citep{eisenberg2001systemic}): i) \emph{absolute priority}, i.e., banks with sufficient assets pay their liabilities in full, ii) \emph{limited liability}, i.e., banks with insufficient assets to pay their liabilities are in default and pay all of their assets to lenders, subject to default costs, and iii) \emph{proportionality}, i.e., in case of default, payments to lenders are made in proportion to the respective liability. Payments that satisfy the above properties are called \emph{clearing payments} and (perhaps surprisingly) these payments are not uniquely defined for a given financial system. However, maximal clearing payments, i.e., ones that point-wise maximize all corresponding payments, are known to exist and can be efficiently computed \citep{rogers2013failure}.

\pk{Importantly, banks that are in default after clearing cannot fully repay their liabilities, which can trigger a cascading effect on the network and lead to a collapse of the financial system. Previous work has considered several systemic risk measures (see, e.g., ~\citep{chen2013axiomatic,kromer2016systemic,feinstein2017measures}), in an attempt to quantify the risk of default contagion.}
\mk{In this paper, we consider the total liquidity of a financial system, i.e. the sum of payments made at clearing, as a natural metric for the well-being of the system~\citep{lee2013systemic,dong2021some}}. Financial authorities, e.g., governments or other regulators, wish to keep the systemic liquidity as high as possible and they might interfere, if their involvement is necessary and would considerably benefit the system. For example, in the not so far past, the Greek government (among others) took loans in order to bailout banks that were in danger of defaulting, to avert collapse. In this work, we study the possibility of a financial regulating authority performing cash injections (i.e., bailouts) to selected bank(s) and/or forgiving debts selectively, with the aim of maximizing the total liquidity of the system (total money flow). Similarly to cash injections, it is a fact that debt removal can have a positive effect on systemic liquidity. Indeed, the existence of default costs can lead to the counter-intuitive phenomenon whereby removing a debt/edge from the financial network might result in increased money flow, e.g., if the corresponding borrower avoids default costs because of the removal. 

Even more surprising than the increase of liquidity by the removal of debts, is the fact that the removal of an edge from borrower $b$ to lender $l$ might result in $l$ receiving more incoming payments, e.g., if $b$ avoids default costs and there is an alternative path in the network where money can flow from $b$ to $l$. This motivates the definition of an edge-removal game on financial networks, where banks act as strategic agents who wish to maximize their total assets and might intentionally give-up a part of their due incoming payments towards this goal. As implied earlier, removing an incoming debt could rescue the borrower from financial default, thereby avoiding the activation of default costs, and potentially increasing the lender's utility (total assets). This strategic consideration is meaningful both in the context where a financial authority performs cash injections or not. We consider the existence, quality, and computation of equilibria that arise in such games.

\subsection{Our contribution}
We consider computational problems related to maximizing systemic liquidity, when a financial authority can modify the network by appropriately removing debt, or by injecting cash into selected banks. We also consider financial network games where agents (corresponding to banks) can choose to forgive incoming debts. 

We show how to compute the optimal cash injection policy in polynomial time when there are no default costs, by solving a linear program; the problem is \NP-hard when non-trivial default costs apply. As our LP-based algorithm requires knowledge of the available budget and leads to non-monotone payments, we study the approximation ratio of a greedy cash injection policy that is budget-agnostic and guarantees monotone payments. Regarding debt removal, we prove that the problem of finding the set of liabilities whose removal maximizes systemic liquidity is \NP-hard, and so are relevant optimization problems.

Regarding edge-removal games, with or without bailout, we study the existence and the quality of Nash equilibria, while also addressing computational complexity questions. Apart from arguing about well-established notions, such as the Price of Anarchy and the Price of Stability, we introduce the notion of the Effect of Anarchy (Stability, respectively) as a new measure on the quality of equilibria in this setting.

\subsection{Related work}
\hz{The interconnectedness of the modern financial system is widely regarded as a key factor contributing to the recent financial crisis. Systemic risk refers to the risk \mk{of failure} that the financial system is prone to,  due to its inherent characteristics. \mk{Various} systemic risk measures have been discussed in the literature (see in, e.g.,~\citep{chen2013axiomatic,kromer2016systemic,feinstein2017measures})}. \hz{\mk{Moreover,} a large body of work has emerged on default contagion in various financial network models, see in, e.g.,~\citep{eisenberg2001systemic,glasserman2015likely,gai2010contagion,acemoglu2015systemic}}. Our model is based on the seminal work of \cite{eisenberg2001systemic} who 
introduced a widely adopted model for financial networks, assuming debt-only contracts and proportional payments. This was later extended by \cite{rogers2013failure} to allow for default costs; \citeauthor{rogers2013failure} prove the existence of maximal clearing payments when non-trivial default costs apply, and present an efficient algorithm to compute them.
\cite{SSB17} show the complexity of finding clearing payments with credit default swaps. Additional features have been since introduced, see e.g., \citep{schuldenzucker2020default} and \citep{papp2020network} for models including credit default swaps.
\cite{ioannidis2022strong} examine clearing solutions' irrationality and approximation strength, while \cite{IoannidisKV23singleton} \mk{explore} the intractability of clearing problems in networks with derivatives and payment priorities. 


When the financial regulator has available funds to bailout each bank of the network, \cite{jackson2020credit} \mk{consider} the minimum bailout budget to ensure systemic solvency and prove that computing it is an \NP-hard problem.
When the financial authority has limited bailout budget, \cite{demange2018contagion} proposes the threat index as a means to determine which banks should receive cash during a default episode and suggests a greedy algorithm for this process. 
\cite{calafiore2024default} introduce a new concept termed default resilience margin that measures the maximum amplitude of assets price fluctuations without generating additional insolvencies.
\cite{egressy2021bailouts} \mk{study} how central banks should decide which insolvent banks to bailout and formulate corresponding optimization problems, 
\mk{while considering the market value as a metric for the welfare of the system. They derive several hardness results for various natural objectives, most of which assume non-trivial default costs.}
\cite{dong2021some} introduce an efficient greedy-based clearing algorithm for an extension of the Eisenberg-Noe model, while also studying bailout policies when banks in default have no assets to distribute. Note that the problem of injecting cash (as subsidies) in financial networks has been studied (in a different context) in microfinance markets \citep{IO18}. 
\hz{\cite{klages2022optimal} investigate the optimal bailout policy 
within a \mk{financial} network model featuring cross-holdings and portfolio overlaps.}
Most relevant to our setting is the paper by \cite{PapachristouK2022stimulus} who consider a setting where there is a fixed amount that can be offered to each bank and prove that finding the optimal \emph{discrete}\footnote{\mk{Note that the term ``discrete'' in \citep{PapachristouK2022stimulus} refers to a cash injection policy where a bank is either allocated a given fixed amount or not, which is different to the integer payments  considered in Theorem \ref{thm:alg-1}a of our paper.}}  bailout scheme is an \NP-hard problem; they design and analyze approximation algorithms, and
prove hardness-of-approximation results.
Furthermore, in \citep{papachristou2023dynamic}, they extend the static Eisenberg-Noe model in a dynamic manner, formulate the optimal cash injection problem as a Markov Decision Process, present how to efficiently compute optimal policies when amounts can come from a range and provide approximation algorithms for the setting where there is a single fixed amount that can be offered to each bank.

\mph{Relevant to debt forgiveness (or removal) is the concept of network/portfolio compression \citep{DER21}, which can be seen as the removal  of a sequence of debts that form a cycle. \cite{schuldenzucker2021portfolio} analyze the impact of compression on social welfare and banks’ incentives to accept compression proposals. \cite{veraart2022does} derives necessary structural conditions under which portfolio compression could be detrimental to the health of the financial system, while 
\cite{amini2023optimal} introduce a formulation of the optimal network compression problem for financial systems and demonstrate that \mk{it is computationally intractable for different compression models.}  
}

\pk{When taking the banks' incentives also into account,} a growing body of recent work considers strategic aspects in financial networks. Departing from the principle of proportionality \citep{eisenberg2001systemic}, \cite{bertschinger2019strategic} frame payments in the event of bankruptcy as a decision-making process. They propose two distinctive payment schemes, namely unit-ranking (or, coin-ranking), and edge-ranking, respectively, and examine the properties of equilibria in games that emerge in each payment scheme. Subsequent research by \cite{kanellopoulos2021financial} delves into priority-proportional payment schemes. \cite{papp2020network} study the banks' incentives for redefining liabilities and donating external assets. In a follow-up study \citep{PW21}, they consider the reversible network model where banks may return from insolvency, and point out that early defaulting can be an advantageous strategy for banks in some cases, while \cite{allouch2023strategic} investigate the behavior of strategic default in a two-period financial network, and show the existence of Nash equilibria of this default game and develop an algorithm to find them. Additionally, \cite{papp2021debt} study the effects of debt swaps in risk mitigation, while \cite{FroeseHWswapping} examine the algorithmic properties of debt swaps with ranking-based clearing. In recent contributions, \cite{BertschingerHKLSW2023fire} study the existence and structure of equilibria in a game-theoretic model of fire sales, while \cite{KanellopoulosKZ2023financial} examine debt transfer games, where banks strategically decide whether to transfer their debt claims. This theoretical study is complemented by empirical game analysis. \cite{hoefer2024algorithms}  study claims trading in financial networks and analyze the structural properties and computational aspects across various forms of claims trading. \cite{zhou2024strategic} analyze prepayment
games in the existence, inefficiency, and computation of resulting equilibria. \cite{tong2024reducing} focus on strategic donations from solvent banks to distressed banks, demonstrating how such donations can potentially mitigate losses and prevent default cascades. Meanwhile, in \citep{tong2024selfishly}, the same group of authors explore the impact of selfish debt cancellations on systemic risk. To the best of our knowledge, our work is the first to investigate a model for financial networks where agents consider forgiving debts as strategic actions.

	
\section{Preliminaries}\label{sec:preliminaries}
A \emph{financial network} $\fn=(V,E)$ consists of a set $V=\{v_1, \dots, v_n\}$ of $n$ banks, where each bank $v_i$ initially has some non-negative \emph{external assets} $e_i$ corresponding to income received from entities outside the financial system.
Banks have payment obligations, i.e., \emph{liabilities}, among themselves. In particular, a \emph{debt contract} creates a liability $l_{ij}$ of bank $v_i$ (the borrower) to bank $v_j$ (the lender); we assume that $ l_{ij} \geq 0$ and $l_{ii}=0$. Note that $l_{ij}>0$ and $l_{ji}>0$ may both hold simultaneously. Also, let $L_i=\sum_j {l_{ij}}$ be the total liabilities of bank $v_i$.  Banks with sufficient funds to pay their obligations in full are called \emph{solvent} banks, while ones that cannot are \emph{in default}. 
Then, the relative liability matrix $\Pi \in \mathbb{R}^{n\times n}$ is defined by 
\begin{equation*}
\pi_{ij}= \left\{
\begin{array}{ll}
l_{ij}/L_i, & \textrm{if } L_i > 0\\
0, & \textrm{otherwise}.\\
\end{array} \right.
\end{equation*}

Let $p_{ij}$ denote the actual payment\footnote{Note that the actual payment need not equal the liability, i.e., the payment obligation.} from $v_i$ to $v_j$; we assume that $p_{ii} =0$. These payments define a payment matrix $\mathbf{P} = (p_{ij})$ with $i,j \in [n]$, where by $[n]$ we denote the set of integers $\{1, \dots, n\}$. We denote by $p_i = \sum_{j \in [n]}{p_{ij}}$ the total outgoing payments of bank $v_i$. A bank in default may need to liquidate its external assets or make payments to entities outside the financial system (e.g., to pay wages). This is modeled using \emph{default costs} defined by values $\alpha, \beta \in [0,1]$. A bank in default can only use an $\alpha$ fraction of its external assets and a $\beta$ fraction of its incoming payments (the case without default costs is captured by $\alpha=\beta=1$).  The absolute priority and limited liability regulatory principles, discussed in the introduction, imply that a solvent bank must repay all its obligations to all its lenders, while a bank in default must repay as much of its debt as possible, taking default costs also into account; the requirement for \emph{proportional payments} dictates how payments happen in the latter case. 
Summarizing, it holds that $p_{ij}\leq l_{ij}$ and, in particular, 
$\mathbf{P}=\Phi(\mathbf{P})$, where


\begin{equation*}
\Phi(\mathbf{x})_{ij}= \left\{
\begin{array}{l}
l_{ij}, \quad \textrm{if } L_i\leq e_i+\sum_{j=1}^n x_{ji}\\
(\alpha e_i+ \beta \sum_{j=1}^n x_{ji}) \cdot\pi_{ij}, \quad \textrm{otherwise}.\\
\end{array} \right.
\end{equation*}

Proportional payments $\mathbf{P}$ that satisfy these constraints are called \emph{clearing payments} and have been frequently studied in the financial literature (e.g., in \citep{eisenberg2001systemic,rogers2013failure,demange2018contagion}). 

Given clearing payments $\mathbf{P}$,  the \emph{total assets} $a_i(\mathbf{P})$ of bank $v_i$ 
are defined as the sum of external assets plus incoming payments, i.e., $$a_i(\mathbf{P})=e_i+\sum_{j\in [n]} p_{ji}.$$ Maximal clearing payments, i.e., ones that point-wise maximize all corresponding payments (and hence total assets), are known to exist \citep{eisenberg2001systemic,rogers2013failure} and can be computed in polynomial time. As is standard practice, we focus on these payment throughout the paper, also to avoid ambiguity.

We measure the total \emph{liquidity} of the system (also referred to as systemic liquidity) $\liq(\mathbf{P})$ as the sum of payments traversing through the network, i.e., $$\liq(\mathbf{P})=\sum_{i\in [n]} \sum_{j\in [n]} p_{ji}.$$ We sometimes omit dependence on $\mathbf{P}$ for clarity of exposition when the payments are clear from the context.

We assume that there exists a financial authority (a \emph{regulator})
who aims to maximize the systemic liquidity. 
In particular, the regulator can decide to remove certain debts (edges) from the network or inject cash to some bank(s). In the latter case, we assume the regulator has a total \emph{budget} $M$ available in order to perform \emph{cash injections} to individual banks. We sometimes refer to the total increased liquidity, $\Delta\liq$, (as opposed to total liquidity) which measures the difference in the systemic liquidity before and after the cash injections.\footnote{This is necessary as in some cases, like the proof of the approximability of the greedy algorithm in Theorem \ref{Greedy algorithm}, we cannot argue about the total liquidity but we can argue about the total increased liquidity.} 
A \emph{cash injection policy} is a sequence of pairs of banks and associated transfers $((i_1, t_1), (i_2, t_2),$ $\ldots (i_L, t_L)) \in (V\times \mathbb{R})^L$, such that the regulator gives capital $t_1$ to bank $i_1$, $t_2$ to bank $i_2$, etc. \mk{We note that $t_k$ can be a fractional number unless stated otherwise.} These actions naturally define two corresponding optimization problems on the total (increased) liquidity, i.e., \emph{optimal cash injection} and \emph{optimal debt removal}.

\begin{opt-cash-inj} [\textsc{Opt-Cash-Injection}] Given a financial network $\fn$ and a total budget $M$, compute the cash injection policy that maximizes the total increased liquidity $\Delta\liq$.
\end{opt-cash-inj}

\begin{opt-debt-rem} [\textsc{Opt-Debt-Removal}] Given a financial network $\fn$, compute a collection of edges whose removal maximizes the total liquidity $\liq$.
\end{opt-debt-rem}
We will also find useful the notion of the \emph{threat index}\footnote{The term threat index aims to capture the ``threat'' posed to the network by a decrease in a bank's cash flow or even the bank's default; this index can be thought of as counting all the defaulting creditors that would follow a potential default of the said bank.}, $\mu_i$, of bank $v_i$, which captures how many units of total increased liquidity will be realized if the financial authority injects one unit of cash into bank $v_i$'s external assets \citep{demange2018contagion}; a unit of cash represents a small enough amount of money so that the set of banks in default would not change after the cash injection. Naturally, the threat index of solvent banks is $0$, while the threat index of banks in default will be at least $1$.  Formally, the threat index is defined as 
\begin{equation*}
\mu_i= \left\{
\begin{array}{ll}
1+\sum_{j \in D}\pi_{ij}\mu_j, & \textrm{if } a_i(\mathbf{P}) < L_i\\
0, & \textrm{otherwise},\\
\end{array} \right.
\end{equation*}
where $D=\left\lbrace j|a_j(\mathbf{P}) < L_j\right\rbrace $ is the set of banks who are in default.\footnote{In matrix form, the threat index for banks who are in default can be computed by $\left( \bm{\mathbb{I}}_{D \times D }-\bm{\pi}_{D \times D}\right)  \bm{\mu}_{D}=\bm{1}_{D}$. This is a homogeneous linear equation system where $\bm{\pi}_{D \times D}$ is the relative liability matrix only involving in the banks in set $D$; $\bm{\mathbb{I}}_{D \times D }$ and $\bm{1}_{D}$ represent the $|D| \times |D|$ dimension identity matrix and the $|D|$ dimension identity vector, respectively, while $\bm{\mu}_{D}$ is the vector of threat indices for banks in default.
} \mk{We remark that for the maximum total increased liquidity it holds $\Delta\liq \leq M\cdot \mu_\texttt{max}$, where $\mu_\texttt{max}$ is the maximum threat index.}
\medskip

\noindent\textbf{An example.} Figure \ref{fig:e1} provides an example of a financial network, inspired by an example in \citep{demange2018contagion}. The clearing payments are as follows: $p_{21}=4.4$, $p_{32} = 3.2$, and $p_{43}=p_{45}=1$, implying that banks $v_2, v_3$ and $v_4$ are in default. We assume that there are no default costs, i.e., $\alpha=\beta=1$.
The threat indexes are computed as follows: $\mu_1=\mu_5 = 0$, $\mu_2 = 1+\mu_1$, $\mu_3 = 1+\mu_2$, and $\mu_4 = 1+\frac{1}{2}\mu_3+\frac{1}{2}\mu_5$, implying that $\mu_3=\mu_4=2$, $\mu_2=1$, while $\mu_ 1=\mu_5=0$. 

 \begin{figure}[htbp]
	\centering
	\includegraphics[scale=0.7]{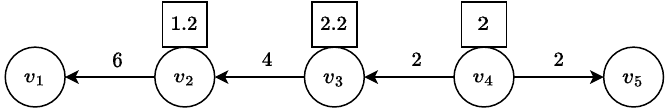}
	\caption{A simple financial network. Nodes correspond to banks, edges are labeled with the respective liabilities, while external assets are in a rectangle above the relevant banks.}
	\label{fig:e1}
\end{figure}

Additional definitions are deferred to the corresponding sections they are used in, for ease of exposition. Most importantly, we introduce (in Section~\ref{sec:games}) two novel notions, namely the Effect of Anarchy/Stability,  that measure the effect strategic behavior has on the original network. 


\section{Computing optimal outcomes}\label{sec:algorithmic}
In this section we present algorithmic and complexity results regarding the problems of computing optimal cash injection (see Section \ref{sec:bailout}) and debt removal (Section \ref{sec:bailin}) policies. 
Note that we omit referring to default costs in our statements for those results that hold when $\alpha=\beta=1$.  

\subsection{Optimal cash injections}\label{sec:bailout}
We begin with a positive result about computing the optimal cash injection policy when default costs do not apply. \hz{\mk{The proof utilizes a linear program which} is a straightforward extension of the \mk{corresponding program for calculating clearing payments in the basic model without cash injections, see e.g. \citep{eisenberg2001systemic}. This has been also observed by}  \cite{demange2018contagion} and \cite{ahn2019optimal} without an explicit focus on the computational complexity aspect.}

\begin{theorem}\label{thm:opt-cash-injection}

\oci{} can be solved in polynomial time.

\end{theorem}




\begin{proof}
The proof follows by solving a linear program that computes the optimal cash injections and accompanying payments. 
	\begin{align*}
\begin{array}{ll@{}ll}
\text{maximize}  & \displaystyle\sum_i{\sum_j{p_{ij}}}& \\
\text{subject to}& \sum_{i}{x_i} \leq M,  & &\forall i, j\\
 &     p_{ij} \leq l_{ij},  & & \forall i, j\\
 &     p_{ij} \leq (x_i+e_i+\sum_{k}{p_{ki}}) \cdot \frac{l_{ij}}{L_i},  & & \forall i, j\\
&      x_i \geq 0,& & \forall i \\
&     p_{ij} \geq 0,& &\forall i,j
\end{array}
	\end{align*}
	
We denote by $x_i$ the cash injection to bank $i$ and by $p_{ij}$ the payment from $i$ to $j$. We aim to maximize the total liquidity, i.e., the total payments, subject to satisfying the limited liability and absolute priority principles.
 Recall that $M$ is the budget, $l_{ij}$ is the liability of $i$ to $j$, $e_i$ is the external assets of bank $i$, and $L_i$ is the total liabilities of $i$.

The first constraint corresponds to the budget constraint, while the second and third sets of constraints guarantee that no bank pays more than her total assets or more than a given liability; hence, the limited liability principle is satisfied. It remains to argue about the absolute priority principle, i.e., a bank can pay strictly less than her total assets only if she fully repays all outstanding liabilities.
	
Consider the optimal solution corresponding to a vector of cash injections and payments $p_{ij}$; we will show that this solution satisfies the absolute priority principle as well. We distinguish between two cases depending on whether a bank is solvent or in default. In the first case, consider a solvent bank $i$, i.e., $x_i+e_i+\sum_j{p_{ji}}\geq L_i$, for which $p_{ik}< l_{ik}$ for some bank $k$. By replacing $p_{ik}$ with $p'_{ik} = l_{ik}$, we obtain another feasible solution that strictly increases the objective function; a contradiction to the optimality of the starting solution. Similarly, consider a bank $i$ with $x_i+e_i+\sum_j{p_{ji}}< L_i$ for which $\sum_{j} p_{ij} < x_i+e_i+\sum_j{p_{ji}}$. Then, there necessarily exists a bank $k$ for which $p_{ik} < (x_i+e_i+\sum_j{p_{ji}})\cdot \frac{l_{ik}}{L_i}$ and it suffices to replace $p_{ik}$ with $p'_{ik} = (x_i+e_i+\sum_j{p_{ji}})\cdot \frac{l_{ik}}{L_i}$ to obtain another feasible solution that, again, strictly increases the objective function. Hence, we have proven that the optimal solution to the linear program satisfies the absolute priority principle and the claim follows by providing each bank $i$ a cash injection of~$x_i$.
\end{proof}

Note that the optimal policy does not satisfy certain desirable properties. In particular, as observed in \citep{demange2018contagion}, cash injections are not monotone with respect to the budget. To see that, consider the financial network in Figure \ref{fig:e1} and note that when $M=0.5$, the optimal policy would give all available budget to bank $v_3$, while under an increased budget of $1.6$, the entire budget would be allocated to $v_4$, hence $v_3$ would get nothing. Furthermore, our LP-based algorithm crucially relies on knowledge of the available budget. In an attempt to alleviate these undesirable properties, we turn our attention to efficiently approximating the optimal cash injection policy by a natural and intuitive greedy algorithm, 
and we compute its approximation ratio under a limited budget, when we care about the total \emph{increased} liquidity\footnote{\hz{Note that the approximation ratio in Theorem \ref{Greedy algorithm} remains valid and serves as a lower bound when we consider total liquidity itself, i.e., $\frac{\liq_{\greedy}}{ \liq_{OPT}}= \frac{\liq^{i}+\Delta\liq_{\greedy}}{ \liq^{i}+\Delta\liq_{OPT}} \geq \frac{\Delta\liq_{\greedy}}{ \Delta\liq_{OPT}}$, where $\liq^{i}$ is non-negative and represents the total liquidity without any cash injections.}}.


\mk{
\begin{algorithm}
\caption{The \greedy\ cash injection policy} 
\begin{algorithmic}[1]
\Require $M>0$ \Comment{Initial budget}
\State $k=0$ \Comment{Let Round $0$ be the initial network}

\While {$M>0$ and at least one bank in default}

\State $k \gets k+1;$

\State $i_k$: the bank with the highest threat index;

\State $t_k$: the minimum amount that,  if transferred to $i_k$, would lead to some previously defaulting bank to become solvent;

\State Inject $\min\{M,t_k\}$ to $i_k$;

\State $M \gets M - \min\{M,t_k\}$
\EndWhile
\end{algorithmic} 
\end{algorithm}}

\begin{definition}[\mk{\greedy\ 's  approximation ratio}]
	
The \emph{approximation ratio} of \greedy\ shows how smaller the total \emph{increased} liquidity (or money flow) can be, compared to the optimal total increased liquidity, and is computed as 
	\begin{equation*}
	\mathcal{R}_{\greedy}=\min_{\fn, M} \frac{\Delta\liq_{\greedy}}{\Delta\liq_{OPT}},
	\end{equation*}
	where the minimum is computed over all possible networks and budgets.
\end{definition}

Let us revisit the example in Figure  \ref{fig:e1}, assuming a budget $M= 1.6$. Initially banks $v_3$ and $v_4$ have the highest threat index  of  $\mu_3= \mu_4=2$  compared to $\mu_1= \mu_5=0$, and $ \mu_2=1$. We can assume\footnote{This is consistent to our tie breaking assumption that favors the least index.} that bank $v_3$ would receive the first cash injection ($i_1=v_3$) and in fact this will be equal to $t_1=0.8$. Indeed, a cash injection of $0.8$ to $v_3$ will result in $v_3$ becoming solvent (notice that $v_3$ receives $1$ from $v_4$), while a smaller cash injection would not impose any change on the threat index vector. At this stage, the threat index of each bank is as follows $\mu_1'=\mu_3'= \mu_5'=0$ and $\mu_2'= \mu_4'=1$. At this round, $i_2=v_2$ would receive the remaining budget of $t_2=0.8$. Hence, the \emph{total increased liquidity} achieved by \greedy\ at this instance is $\Delta\liq_{\greedy}= 
2.4$ ($t_1$ will traverse edges $(v_3,v_2)$ and $(v_2,v_1)$, while $t_2$ will traverse edge $(v_2,v_1)$). However, the optimal cash injection policy is to inject the entire budget $M=1.6$ to bank $v_4$ resulting in $\Delta\liq_{OPT}=3.2$. Therefore, this instance reveals $\mathcal{R}_{\greedy}\leq \frac{2.4}{3.2}=3/4$.

\begin{theorem}\label{thm:greedy_polynomial}
Algorithm \greedy{} runs in polynomial time.    
\end{theorem}
\begin{proof}
\mk{We note that in each round of \greedy\ at least one (previously defaulted) bank will become irrevocably solvent, so \greedy\ will terminate in at most $n$ rounds.} 

We now show any single round $k$ of \greedy\ takes polynomial time; clearly, this suffices for the result to follow. For a given network $G^k$ at the beginning of the $k$-th round of \greedy, we can run the polynomial-time clearing algorithm proposed in~\citep{rogers2013failure} to determine the set of defaulted banks $D^k$ and the total unpaid liabilities $\hat{L}_{j_k} = \sum_{i \in N} (l_{j_k,i}-p_{j_k, i})$ for each defaulted bank $j_k \in D^k$. In addition, the threat index for each bank in $G^k$ can be computed in polynomial time by solving a system of homogeneous linear equations~\citep{demange2018contagion}. The bank $i_k$, which has the highest threat index, can then be straightforwardly selected. Furthermore, we can always find a sufficiently small value $\hat{t}_k$ of cash injection  to bank $i_k$, such that the set of defaulted banks $D^k$ remains unchanged after the injection. \pk{Indeed, a value \pk{$\hat{t}_k < \frac{\min_{j_k \in D^k}{\hat{L}_{j_k}}}{\max_{i}{\mu_i}}$} suffices as the additional payments cannot change $D^k$.} Let $\delta_{i_k,j_k}$ represent the increase in incoming payments that bank $j_k \in D^k$ receives after this injection, which can be computed by rerunning the polynomial-time clearing algorithm. The ratio $\frac{\delta_{i_k,j_k}}{\hat{t}_k}$ measures the marginal increase in incoming payments that bank $j_k$ receives due to the cash injection to bank $i_k$. Therefore, the cash injection $t_k$ in \greedy\ can be computed as 

 \begin{equation*}
     t_k=\min_{j_k \in D^k}\frac{\hat{L}_{j_k} \cdot \hat{t}_k}{\delta_{i_k,j_k}}
 \end{equation*}
  where $\frac{\hat{L}_{j_k} \cdot \hat{t}_k}{\delta{i_k,j_k}}$ represents the minimum value of cash injection required for bank $i_k$ such that the increased incoming payments to bank $j_k$ can cover its not-yet-paid liabilities, making it solvent and thus changing the set of defaulted banks $D^k$. Clearly, computing $t_k$ for any given round $k$ takes polynomial time; this concludes the proof. 
\end{proof}








\begin{theorem}\label{Greedy algorithm}
	\greedy's approximation ratio is at most $3/4$. For inputs satisfying $M\leq t_1\frac{\mu_{v}}{\mu_{v}-1}$, this  ratio is tight.
\end{theorem}
\begin{proof}
The upper bound follows from the instance in Figure~\ref{fig:e1} and the discussion above; in the following, we argue about the lower bound when the total budget $M$ satisfies the condition in the statement. We consider the following properties and claim that the worst approximation ratio of \greedy\ is achieved at networks that satisfy both properties. To prove this claim we will show that starting from an arbitrary network $\fn$ on which \greedy\ approximates the optimal total increased liquidity by a factor of $r$, we can create a new network that satisfies properties (P1) and (P2), such that \greedy\ approximates the optimal total increased liquidity in the new network by a factor of at most $r$. We can then bound the approximation ratio of \greedy\ on the set of networks that satisfy these properties. 
\begin{itemize}
	\item[(P1)] The total increased liquidity achieved by \greedy\ is exactly $t_1 \mu_{i_1}+(M-t_1)$.
	\item[(P2)] The optimal total increased liquidity is exactly $\mu_{i_1} M$.
\end{itemize}

Let the bank with the highest threat index in $\fn$ be $v$ and let $\mu_{v}$ be that threat index. If $\mu_{v}=1$ or $M\leq t_1$ then the claim is true (\greedy\ trivially achieves an approximation ratio of $1$), so we henceforth assume that $\mu_{v}>1$ and $M>t_1$. We create $\fn'$ (see Figures \ref{fig:Greedy_approx-ub} and  \ref{fig:Greedy_approx-bad_network}) as follows. To formally describe our construction it is convenient to express the highest threat index as $\mu_{v}=x+\frac{a}{a+b}$ for some $a,b>0$ that satisfy $a+b=t_1$ and $x$ a fixed integer. If $\mu_{v}$ is an integer, then we set $x=\mu_{v}-1$ (which implies $a=t_1$ and $b=0$), while if $\mu_{v}$ is not an integer, we set $x=\lfloor \mu_{v}\rfloor$ and $a=(\mu_{v}-\lfloor \mu_{v}\rfloor) t_1$.  Our network $\fn'$ comprises $\lceil \mu_{v}\rceil+4$ nodes, i.e., $u,v,w,z$ and $v_i$, for $i=1, \ldots, {\lceil \mu_{v}\rceil}$. There is a directed path of length $\lceil \mu_{v}\rceil$, including the following nodes in sequence $v$ and $v_i$, for $i=1, \ldots, {\lceil \mu_{v}\rceil}$, where the first $\lfloor \mu_{v}\rfloor-1$ edges have liability $t_1$. The remaining edge(s) on that path has/have liability $a$ and there is an edge $(v_{\lfloor \mu_{v}\rfloor-1},u)$ with liability $b$. Moreover, there are edges $(w,v)$ and $(w,z)$ such that $l_{wv}=t_1$ and ${l_{wz}}=\frac{t_1}{(\mu_{v}-1)}$. Note that the liabilities of the two outgoing edges of $w$ are selected so that $w$ and $v$ both have the highest threat index in $\fn'$; indeed, the threat index of $v$ in $\fn'$ is $\mu_{v}'=\mu_{v}$, and the threat index of $w$ in $\fn'$ is $\mu_{w}'=1+\frac{l_{wv}}{L_w}\mu_{v}=\mu_{v}$, where $L_w=l_{wv}+l_{wz}$. Immediate consequences of our construction are (i) $\mu_{v}=\frac{L_w}{l_{wz}}$ and (ii) $\frac{l_{wv}}{L_w}=\frac{\mu_{v}-1}{\mu_{v}} $ which will be useful later. 

\begin{figure}[htbp]
\centering
\includegraphics[scale=0.6]{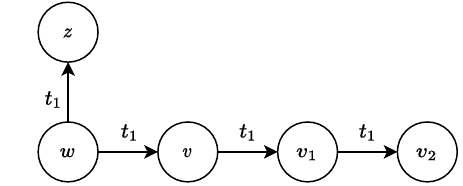}
\caption{An example network used in the lower bound of the approximation ratio of \greedy\ for $\mu_{v}=2$ and $\mu_{w}=1+\frac{1}{2}\cdot 2=2$. The claim in the proof is that for any arbitrary network such that the first cash injection made by \greedy\ is $t_1$, and the highest threat index is an integer, e.g. $2$ in this case, the network in this figure achieves at most the same approximation ratio, while satisfying properties (P1) and (P2).}
\label{fig:Greedy_approx-ub}
\end{figure}

\begin{figure}[htbp]
\centering
\includegraphics[scale=0.6]{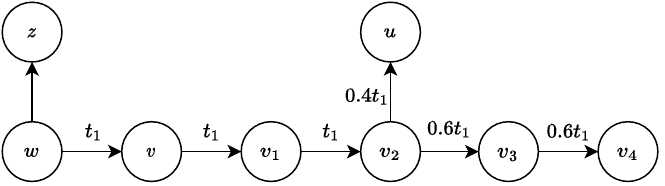}
\caption{Similarly to Figure \ref{fig:Greedy_approx-ub}, for the case where the highest threat index is not an integer, e.g. $\mu_{v}=3.6=1+(1+(1+\frac{0.6t_1}{t_1}\cdot 1))$ and $\mu_{w}=1+\frac{t_1}{t_1+t_1/2.6}\cdot \mu_{v}=1+\frac{2.6}{3.6}\cdot 3.6=3.6$.}
\label{fig:Greedy_approx-bad_network}
\end{figure}

To see that $\fn'$ satisfies property (P1), it suffices to consider that in the first step, \greedy\ offers a cash injection of $t_1$ to $v$ in $\fn'$; $w$ is then the only node in default and has threat index equal to $1$. The total increased liquidity achieved by \greedy\ in $\fn'$ is exactly $t_1\mu_{v}+(M-t_1)$, while \greedy\ achieves at least that total increased liquidity in $\fn$ since each node in default has threat index at least equal to $1$. Assume now that the regulator offers the entire  budget to node $w$ in $\fn'$; this would result to total increased liquidity of $\mu'_{w} M$ as required by (P2), and since, by construction $\mu'_{w} =\mu_{v}$ is the maximum total increased liquidity in $\fn$, it holds that $\mu'_{w} M$ is an upper bound on the optimal total increased liquidity in the original network too. We can conclude that it is without loss of generality to restrict attention to networks that satisfy properties (P1) and (P2) when proving a lower bound on the approximation ratio of {\greedy}.

Overall, it holds that the approximation ratio of \greedy\ is lower bounded by

\begin{align*}
\mathcal{R}_{\greedy}&\geq\min_{\mathcal{\fn}^*, M}\frac{\Delta\liq_{\greedy}}{\Delta\liq_{OPT}}\\
& \geq \min_{\mathcal{\fn}^*, M}\frac{t_1\mu_{v}+(M-t_1)}{M \mu_{v}},
\end{align*}
where the minimum ranges over the set of networks $\mathcal{\fn}^*$ satisfying properties (P1) and (P2). Straightforward calculations after substituting $\mu_{v}=\frac{L_w}{l_{wz}}$ which holds by construction, leads to
\begin{align*}
\mathcal{R}_{\greedy}&\geq \min_{\mathcal{\fn}^*, M} \left\{1+\left(\frac{t_1}{M}-1\right)\frac{l_{wv}}{L_w}\right\}\\
&\geq \min_{\mathcal{\fn}^*}\left\{ 1+\left(\frac{\mu_{v}-1}{\mu_{v}}-1\right)\frac{\mu_{v}-1}{\mu_{v}}\right\}\\
&\geq 1-1/4\\
&=3/4,
\end{align*}
where the second inequality holds by assumption that $M\leq t_1\frac{\mu_{v}}{\mu_{v}-1}$ and since, by construction $\frac{l_{wv}}{L_w}=\frac{\mu_{v}-1}{\mu_{v}} $, while the third inequality holds since $-1/4$ is a global minimum of function $f(x)=(x-1)x$. 
\end{proof}

We conclude this section with some computational complexity results. \mk{We note the contrast to similar results in the literature, which apply to slightly different models, e.g. ``binary'' bailout decisions \hz{in~\citep{PapachristouK2022stimulus}}, or ensuring systemic solvency, i.e., all agents become solvent, \hz{in~\citep{jackson2020credit}}.}
\begin{theorem}\label{thm:alg-1}
The following problems are \emph{\NP}-hard: 
	\begin{enumerate}[label=\emph{\alph*})]

  \item  \oci{} under the constraint of integer payments.

  \item   \oci{} with default costs $\alpha \in [0,1)$ and $\beta \in [0,1]$.

  \item Compute the minimum budget so that a given agent becomes solvent, with default costs $\alpha \in [0,1/2)$ and $\beta \in [0,1]$.
	\end{enumerate}
\end{theorem}
\begin{proof}
We begin with the case where all proportional payments need to be integers, and then prove the cases where default costs apply.
\medskip


\textbf{Hardness of computing the optimal cash injection policy under integer payments.} 
We warm-up with a rather simple reduction from \textsc{Exact Cover by 3-Sets} (X3C), a well-known \NP-complete problem. An instance of X3C consists of a set $X$ of $3k$ elements together with a collection $C$ of size-3 subsets of $X$. The question is whether there exists a subset $C' \subseteq C$ of size $k$ such that each element in $X$ appears exactly once in $C'$.

 We build an instance of our problem as follows. We add an agent $u_i$ for each $c_i \in C$ and an agent $t_i$ for each element $x_i \in X$; we also add agent $T$. There are no external assets and the liabilities are as follows. Each agent $u_i$, corresponding to $c_i = \{a, b, c\}$ where $a,b,c$ are elements in $X$, has liability $1$ to each of the three agents $t_a, t_b, t_c$. Furthermore, each agent $t_i$ has liability of $1$ to agent $T$, while we assume that the budget equals $3k$. 

Due to the integrality constraint and the fact that payments are proportional, a yes-instance for X3C, admitting a solution $C'$, leads to a solution of liquidity $6k$, by injecting a payment of $3$ to each $u_i$ corresponding to the set $c_i \in  C'$. 

On the other direction, we argue that any solution with liquidity at least $6k$ leads to a solution for the instance of X3C. Indeed, any such solution must necessarily lead to liquidity of at least $3k$ to the edges from the $u_i$ agents to the $t_i$ agents. Since the budget equals $3k$ this implies that exactly $k$ of the $u_i$ agents must receive a payment of $3$ and these agents should cover the entire set of the $t_i$ agents.

\medskip

\textbf{Hardness of computing the optimal cash injection policy with default costs.}
Our proof follows by a reduction from the \textsc{Partition} problem, a well-known \NP-complete problem. Recall that in \textsc{Partition}, an instance $\mathcal{I}$ consists of a set $X$ of positive integers $\{x_1, x_2, \dots, x_k\}$ and the question is whether there exists a subset $X'$ of $X$ such that $\sum_{i \in X'}{x_i} = \sum_{i \notin X'}{x_i} = \frac{1}{2}\sum_{i \in X}{x_i}$.

The reduction works as follows. Starting from $\mathcal{I}$, we build an instance $\mathcal{I}'$ by adding an agent $v_i$ for each element $x_i \in X$ and allocating an external asset of $e_i = x_i$ to $v_i$; we also include three additional agents $S, T$ and $G$. Each agent $v_i$ has liability equal to $\frac{4e_i}{3}$ to $S$ and equal to $\frac{2e_i}{3}$ to $T$, while $S$ has liability $\frac{2+\alpha}{3}\sum_{i}{e_i}$ to $G$; see also Figure \ref{Optimal bailout with default cost}. We assume the presence of default costs $\alpha \in [0,1)$, and $\beta \in [0,1]$, while the budget is $M = \frac{1}{2}\sum_i{e_i}$; clearly, the reduction requires polynomial-time. 

\begin{figure}[htbp]
	\centering
	\includegraphics[scale=0.63]{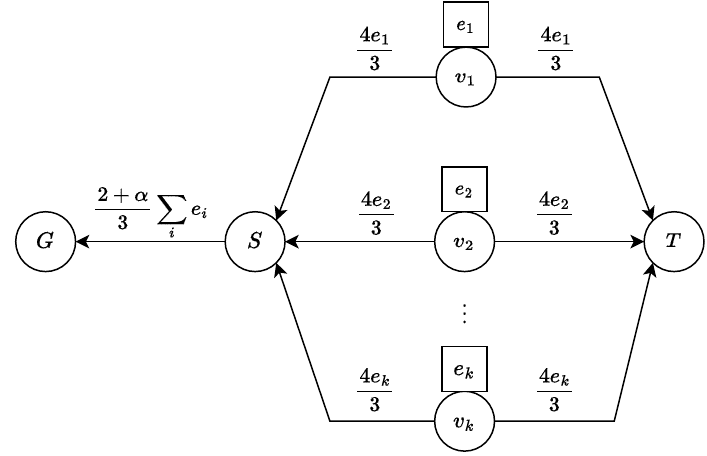}
	\caption{The reduction used to show hardness of computing the optimal cash injection policy when $\alpha<1$.}
	\label{Optimal bailout with default cost}
\end{figure}

We first show that if $\mathcal{I}$ is a yes-instance for \textsc{Partition}, then the total liquidity is $\liq = \frac{5\alpha+10}{6} \sum_i e_i$. Indeed, consider a solution $X'$ for instance $\mathcal{I}$ satisfying $\sum_{i \in X'}{x_i} = \frac{1}{2}\sum_{i \in X}{x_i}$, and let the set $B$ contain agents $v_i$ where $x_i \in X'$. Then, since $M = \frac{1}{2}\sum_i{e_i}  = \sum_{i \in B}{e_i}$, we choose to inject an amount of $e_i$ to any agent $v_i \in B$. The total assets of agent $S$ are then 
\begin{align*}
a_S &= \frac{4}{3}\sum_{i \in B}{e_i} +\frac{2 \alpha}{3}\sum_{i \notin B}{e_i}\\
&= \frac{2+\alpha}{3}\sum_i{e_i},
\end{align*}
and, hence, $S$ is solvent. The total liquidity in this case is  
\begin{align*}
\liq &= 2\sum_{i \in B}{e_i} + \alpha \sum_{i \notin B}{e_i} + \frac{2+\alpha}{3}\sum_i{e_i}\\
&= \frac{5\alpha+10}{6} \sum_i e_i,
\end{align*}
as desired.

We now show that any cash injection policy that leads to a total liquidity of at least $\frac{5\alpha+10}{6} \sum_i e_i$ leads to a solution for instance $\mathcal{I}$ of \textsc{Partition}. Assume any such cash injection policy and let $B$ be the set of $v_i$ agents that become solvent by it. 
Denote by $\liq_{v}$ the liquidity arising solely from  the payments made by the  agents $v_i$, for $i=1,\ldots,k$, to their direct neighbors, and we denote by $t_S$ the amount of cash injected to $S$, we get that  
		\begin{align}\label{eq:liquidity_v_is-lower}\nonumber
\liq_{v}&=2\sum_{i \in B}e_i+\left( \sum_{i \notin B}e_i+\frac{\sum_i e_i}{2}-\sum_{i \in B}e_i-t_S\right) \cdot \alpha\\\nonumber
&=2\sum_{i \in B}e_i+\left(\frac{3}{2}\sum_i e_i-2\sum_{i \in B}e_i-t_S\right)\cdot \alpha\\
&\leq 2(1-\alpha)\sum_{i \in B}e_i+\frac{3\alpha}{2}\sum_i e_i.
		\end{align}

It also holds that 
\begin{equation}\label{eq:liquidity_v_is-upper}
\liq_{v}\geq \left(1+\frac{\alpha}{2}\right)\sum_i e_i
\end{equation}
since, by assumption, the policy under consideration leads to a total liquidity of at least $\frac{5\alpha+10}{6} \sum_i e_i$ and the payment from $S$ to $G$ can be at most $\frac{2+\alpha}{3}\sum_i e_i$. 

Combining Inequalities (\ref{eq:liquidity_v_is-lower}) and (\ref{eq:liquidity_v_is-upper}), we get that
\begin{align*}
2(1-\alpha)\sum_{i \in B}e_i+\frac{3\alpha}{2}\sum_i e_i &\geq \left(1+\frac{\alpha}{2}\right)\sum_i e_i, \qquad \iff\\
2(1-\alpha)\sum_{i \in B}e_i&\geq \left(1-\alpha\right)\sum_i e_i,
\end{align*}
which is equivalent to $\sum_{i \in B}e_i \geq  \frac{1}{2}\sum_i{e_i}$ for $a<1$. Moreover, note that each agent $v_i \in B$ needs (at least) an extra $e_i$ to become solvent and, hence, it must be $\sum_{i \in B}{e_i}\leq \frac{1}{2}\sum_i{e_i}$, due to the budget constraint. We can conclude that $\sum_{i \in B}e_i =  \frac{1}{2}\sum_i{e_i}$ and, hence, we can obtain a solution to instance $\mathcal{I}$ of \textsc{Partition}, as desired.
\medskip

\textbf{Hardness of computing the minimum budget that makes an agent solvent.} 
The proof follows by the same reduction as in the previous case. We will prove that computing the minimum budget necessary to make agent $S$ solvent corresponds to solving an instance from \textsc{Partition}.
As before, whenever instance $\mathcal{I}$ admits a solution $X'$, we inject an amount of $e_i$ to each agent $v_i$ such that  $x_i \in X'$, and, as in the previous case, we obtain that a budget of $\frac{1}{2}\sum_i{e_i}$ suffices to make  $S$ solvent. We now argue that any cash injection policy with a budget of $\frac{1}{2}\sum_i{e_i}$ that can make agent $S$ solvent leads to a solution for instance~$\mathcal{I}$ when the default costs are $\alpha \in [0,1/2), \beta \in [0,1]$.

Let $t_i$ be the cash injected at agent $v_i$ and let $t_S$ be the cash injected directly at agent $S$; clearly, $t_S+\sum_i{t_i}\leq \frac{\sum_i{e_i}}{2}$. As before, let $B$ be the set of $v_i$ agents that become solvent by the cash injection policy. Clearly, if $\sum_{i\in B}{e_i} = \frac{1}{2}\sum_i{e_i}$, we immediately obtain a solution to the \textsc{Partition} instance. Otherwise, $\sum_{i \in B}{e_i}< \frac{1}{2}\sum_i{e_i}$ and the total assets of $S$ are 

	\begin{align*}
		a_S &= \frac{4}{3}\sum_{i \in B}{e_i}+\frac{2 \alpha}{3}\sum_{i \notin B}{(e_i+t_i)}+ t_S\\
		&\leq \frac{4}{3}\sum_{i \in B}{e_i}+\frac{2 \alpha}{3}\sum_{i \notin B}{e_i}+ t_S+\sum_{i \notin B}{t_i}\\
		&= \frac{4}{3}\sum_{i \in B}{e_i}+\frac{2 \alpha}{3}\sum_{i \notin B}{e_i} + \frac{\sum_i{e_i}}{2}-\sum_{i \in B}{e_i}\\
		&= \frac{1}{3}\sum_{i \in B}{e_i}+\frac{1}{2}\sum_i{e_i}+\frac{2 \alpha}{3}\sum_{i \notin B}{e_i}\\
		&= \frac{1}{3}\sum_{i \in B}{e_i}+\frac{1}{2}\sum_i{e_i}+\frac{2 \alpha}{3}\left( \sum_i e_i-\sum_{i \in B} e_i\right) \\
		&=\left( \frac{1}{2}+\frac{2\alpha}{3}\right) \sum_i e_i+\frac{1-2\alpha}{3}\sum_{i \in B}e_i\\
		&< \left( \frac{1}{2}+\frac{2\alpha}{3}\right) \sum_i e_i+\frac{1-2\alpha}{3}\left( \frac{\sum_i e_i}{2}\right) \\
		&= \frac{2}{3}\sum_i{e_i}+\frac{ \alpha}{3}\sum_i{e_i},
	\end{align*}
where the second equality holds due to the budget constraint and the strict inequality holds since $\alpha<1/2$ and $\sum_{i \in B}{e_i}< \frac{1}{2}\sum_i{e_i}$; the claim follows. 
\end{proof}

\subsection{Optimal debt removal}\label{sec:bailin}

In this section, we focus on maximizing systemic liquidity by appropriately removing edges/debts. As an example, consider again Figure \ref{fig:e1}, where the central authority can increase systemic liquidity by removing the edge between $v_4$ and $v_5$.


    

\begin{theorem}\label{thm:alg-3}
\odr{} is \emph{\NP}-hard.
\end{theorem}
\begin{proof}
The proof relies on a reduction from the \NP-complete problem RXC3 \citep{Gon85}, a variant of \textsc{Exact Cover by 3-Sets} (X3C). In RXC3, we are given an element set $X$, with $|X| = 3k$ for an integer $k$, and a collection $C$ of subsets of $X$ where each such subset contains exactly three elements. Furthermore, each element in $X$ appears in exactly three subsets in $C$, that is $|C| = |X| = 3k$. The question is if there exists a subset $C' \subseteq C$ of size $k$ that contains each element of $X$ exactly once. Given an instance $\mathcal{I}$ of RXC3, we construct an instance $\mathcal{I}'$ as follows. We add an agent $t_i$ for each element $i$ of $X$ and an agent $s_i$ for each subset $i$ in $C$, as well as two additional agents, $S$ and $T$. Each $s_i$, corresponding to set $(x, y, z) \in C$, has external assets $e_i=4$ and liability $1$ to the three agents $t_x, t_y, t_z$ corresponding to the three elements $x, y, z \in X$. Furthermore, each $s_i$ has liability $Z$ to agent $S$, where $Z$ is a large integer. Finally, each agent $t_i$ has liability $1$ to agent $T$; see also Figure \ref{fig:10}. Note that this construction requires polynomial time.

\begin{figure}[htbp]
\centering
\includegraphics[scale=0.65]{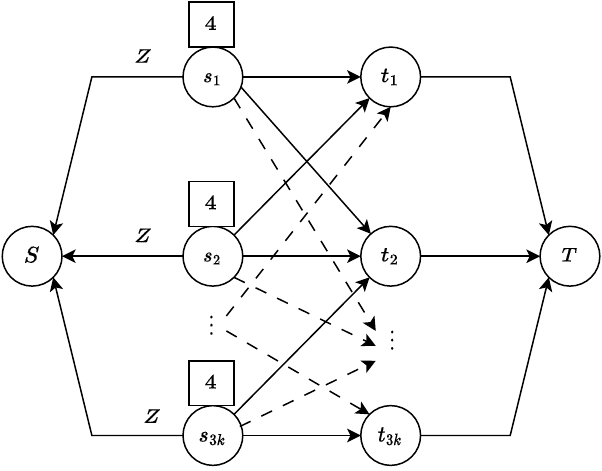}
\caption{The reduction used to show hardness of computing an edge-removal policy that maximizes systemic liquidity. All edges with missing labels correspond to liability $1$.}
\label{fig:10}
\end{figure}

When instance $\mathcal{I}$ is a yes-instance with a solution $C'$, we claim that $\mathcal{I}'$ admits a solution with systemic liquidity $14k$. Indeed, it suffices to remove all edges from the $s_i$ agents, with $i \in C'$, towards $S$. Then, the liquidity due to agents $s_i$, with $i \in C'$, equals $6k$, while each of the $2k$ agents whose edge towards $S$ was preserved generates a liquidity of $4$. 

It suffices to show that any solution that generates liquidity at least $14k$ can lead to a solution for instance $\mathcal{I}$. 
First, observe that it is never strictly better for the financial authority to remove an edge from some agent $s_i$ towards an agent $t_j$. Let $S_k$ and $S_r$ be the subsets of agents $s_i$ whose edges towards $S$ are kept and removed, respectively, and let $\chi = |S_r|$. The liquidity traveling from agents in $S_k$ towards their direct neighbors is exactly $4(3k-\chi)$, while the liquidity traveling from agents in $S_r$ towards their direct neighbors is exactly $3\chi$. The maximum liquidity traveling from agents $t_j$ towards $T$ is at most $\min\{3k, 3\chi + (3k-\chi)\frac{12}{Z+3}\}$. We conclude that the maximum liquidity is bounded by $12k-\chi+\min\{3k, 3\chi + (3k-\chi)\frac{12}{Z+3}\}$. 

Note that, whenever $\chi > k$, then the maximum liquidity is at bounded by $15k-\chi < 14k$. Similarly, whenever $\chi < k$ and since $Z$ is arbitrarily large, the maximum liquidity is bounded by $12k+2\chi + (3k-\chi)\frac{12}{Z+3} < 14k$. It remains to show that whenever $\chi=k$, a liquidity of at least $14k$ necessarily leads to a solution in $\mathcal{I}$. Indeed, by the discussion above, any such solution must have liquidity equal to $3k$ traveling from agents $t_j$ towards $T$, i.e., all these liabilities are fully repaid. This, in turn, can only happen if each of the $t_j$ agents receives a payment of at least $1$ from the $s_i$ agents. Using the assumptions that i) $\chi = k$, ii) $Z$ is arbitrarily large, iii) payments are proportional to liabilities, and iv) each $t_j$ has exactly three neighboring $s_i$ agents, this property holds only when the neighbors of the $\chi$ agents in $S_r$ are disjoint. This directly translates to a solution for instance $\mathcal{I}$ and the RXC3 problem.
\end{proof}



    

We note that the objective of systemic solvency, i.e., guaranteeing that all agents are solvent, can be trivially achieved by removing all edges. However, adding a liquidity target, makes this problem more challenging.

\begin{theorem} \label{thm:bailin-default-cost}
	In networks with default costs, the following problems are \emph{\NP}-hard:
		\begin{enumerate}[label=\emph{\alph*})]
\item \odr{} when required to ensure systemic solvency

\item Compute an edge set whose removal ensures systemic solvency and minimizes the amount of deleted liabilities. 
\item Compute an edge set whose removal guarantees that a given agent is no longer in default and minimizes the amount of deleted liabilities.
	\end{enumerate}
\end{theorem}
\begin{proof}
The proof for all these claims relies on a reduction from the \textsc{Subset Sum} problem, where the input consists of a set $X$ of integers $\{x_1, x_2, \dots, x_k\}$ and a target $t$ and the question is whether there exists a subset $X'\subseteq X$ such that $\sum_{i \in X}{x_i} = t$. Given an instance $\mathcal{I}$ of \textsc{Subset Sum}, we construct an instance $\mathcal{I}'$ by adding an agent $v_i$ for each integer $x_i$, adding an extra agent $v_0$ having $e_0 = t$ and setting the liability of $v_0$ to $v_i$ to be equal to $x_i$; see also Figure~\ref{hardness systemic solvent}. 

\begin{figure}[htbp]
\centering
\includegraphics[scale=0.65]{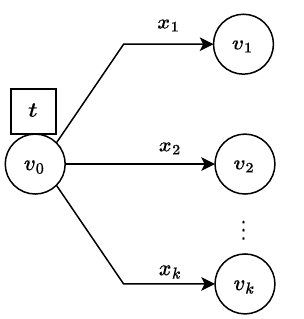}
\caption{An example of the reduction in the proof of Theorem \ref{thm:bailin-default-cost}.}
\label{hardness systemic solvent}
\end{figure}
Since $v_0$ is in default, the goal becomes to remove an edge set so that the remaining liability of $v_0$ is at most $t$. Whenever instance $\mathcal{I}$ is a yes-instance for \textsc{Subset Sum} admitting a solution $X'$, we remove edges from $v_0$ to agents corresponding to integers not in $X'$. Then, $v_0$ is solvent and the systemic liquidity equals $t$. Otherwise, if $\mathcal{I}$ is a no-instance, no edge set removal that leaves $v_0$ solvent can lead to systemic liquidity of (at least) $t$.  
\end{proof}

\section{Forgiving debt as a strategic action}\label{sec:games}
In this section, we consider the case of strategic agents who have the option to forgive debt. We begin with some definitions.
\medskip

\textbf{\mk{Edge-removal games.}} Consider a financial network $\fn$ of $n$ banks who act strategically. \hz{Each bank $v_i$ has the option to select a set of borrowers and forgive the corresponding debts they owe. More formally,} the strategy set of a bank is the power set of its incoming edges and a strategy, $s_i$, denotes which of its incoming edges that bank will remove, thus erasing the corresponding debt owed to itself. A \emph{pure strategy profile}, denoted by $\mathbf{s} = (s_1, \dots, s_n)$ 
 specifies a pure strategy for each bank\hz{; the pure strategy profile for all banks except $v_i$ is denoted by $\mathbf{s}_{-i}$}. 
 The edge-removal game can be defined with and without cash-injections. A given strategy profile will result in realized payments through maximal clearing payments 
 (as discussed in Section \ref{sec:preliminaries}),
 including possible cash injections through a predetermined cash injection policy. 
Our results hold for both the optimal cash injection policy and \greedy{}. 

A bank is assumed to strategize over its incoming edges in order to maximize its utility, i.e., its total assets, where we remark that a possible cash injection can be seen as increasing one's external assets.
The objective of the financial authority is to maximize the total liquidity of the system, i.e., the social welfare is the sum of money flows that traverse the network.
We will now define the notion of \emph{Nash equilibrium} in an edge-removal game. Given a strategy profile $\mathbf{s}$ and the corresponding maximal clearing payments  $\mathbf{P}$, \pk{the utility, i.e., the total assets, of bank $v_i$ is $u_i(\mathbf{P}) = e_i+\sum_{j \in [n]}{p_{ji}}$. We} say $\mathbf{s}$ is a pure Nash equilibrium if for each bank $v_i$ with $i \in [n]$ and any alternative strategy $s_i'$ of $v_i$, it holds that $u_i(\mathbf{P}) \geq u_i (\mathbf{P}')$, where $\mathbf{P}'$ denotes the maximal clearing payments under the profile $\mathbf{s} = (s_i',\mathbf{s}_{-i})$, where player $v_i$ chooses strategy $s_i$ while all other players keep using their equilibrium strategy. . We denote the set of all Nash equilibria in a given edge-removal game by $\mathbf{s}_{eq}$.

 
\textbf{\mk{The Price of Anarchy/Stability.}} The inefficiency of Nash equilibria in terms of liquidity is measured by the \emph{Price of Anarchy} (\emph{Stability}, respectively) \citep{koutsoupias1999worst}, which is equal to the optimal systemic liquidity over that of the worst (best, respectively) pure Nash equilibrium.
Note that the optimal systemic liquidity, $\liq_{OPT}$, corresponds to the maximal one when the financial authority can dictate everyone's actions (edge-removals). 

  \begin{equation*}
 \poa=\max_{\fn} \max_{\mathbf{s}\in \mathbf{s}_{eq}} \frac{\liq_{OPT}}{\liq_{\mathbf{s}}} \qquad  \pos=\max_{\fn} \min_{\mathbf{s}\in \mathbf{s}_{eq}} \frac{\liq_{OPT}}{\liq_{\mathbf{s}}}.
 \end{equation*}


\textbf{\mk{The Effect of Anarchy/Stability.}} We introduce a new notion that we use to measure the discrepancy between the systemic liquidity of the original network (no edge removal) and that of worst (best) Nash equilibrium. We call this the \emph{Effect of Anarchy, (Stability, respectively)} and \mk{use it as an indication of the (positive or negative) effect that strategic behaviour can have on the original network, in contrast to the ``Price of'' measure which indicates the extent to which the individual objectives of the banks and the objective of the regulator are (not) aligned.  The Effect of Anarchy, (Stability, respectively) is defined} as follows. For a given network $\fn$, we have 

  \begin{equation*}
 \eoa(\fn)= \max_{\mathbf{s}\in \mathbf{s}_{eq}} \frac{\liq_{\fn}}{\liq_{\mathbf{s}}} \qquad  \eos(\fn)= \min_{\mathbf{s}\in \mathbf{s}_{eq}} \frac{\liq_{\fn}}{\liq_{\mathbf{s}}}.
 \end{equation*}

The Effect of Anarchy (Stability, respectively) of a given class of networks is determined by the Effect of Anarchy (Stability, respectively) of the network in that class that demonstrates the maximum multiplicative distance of the worst (best, respectively) equilibrium from the original network in terms of liquidity, call that the \emph{extreme} network. The Effect of Anarchy (Stability, respectively) of a given class of networks can be positive or negative depending on whether the respective equilibria admit better or worse liquidity than the original network. 

The Positive Effect of Anarchy (Stability, respectively) of a given class of networks is equal to the inverse of  the Effect of Anarchy (Stability, respectively) of the extreme network  among all networks in the class whose worst (best, respectively) equilibrium is better than the original in terms of liquidity. For example, if a network in a given class has liquidity equal to $1$ with no edge-removals, and the worst (best, respectively) equilibrium admits total liquidity equal to $2$, we say that that class of networks has a Positive Effect of Anarchy (Stability, respectively) of at least  $2$. 

The Negative Effect of Anarchy (Stability, respectively) of a given class of networks is equal to the Effect of Anarchy (Stability, respectively) of the extreme network  among all networks in the class whose worst (best, respectively) equilibrium is weakly worse than the original in terms of liquidity.
For example, if a network in a given class has liquidity equal to $1$ with no edge-removals, and the worst (best, respectively) equilibrium admits total liquidity equal to $1/2$, we say that that class of networks has a Negative  Effect of Anarchy (Stability, respectively) of at least $2$.  

We investigate properties of Nash equilibria in the edge-removal game with respect to their existence and quality, while we also address computational complexity questions under different assumptions on whether default costs and/or cash injections apply. Our results on the Effect of Anarchy of edge removal games imply that, rather surprisingly, in the presence of default costs even the worst Nash equilibrium can be arbitrarily better than the original network in terms of liquidity. However, the situation is reversed in the absence of default costs, where we observe that the original network can be considerably better in terms of liquidity than the worst equilibrium; in line with similar Price of Anarchy results.
We begin with some results for the basic case, that is, without default costs; recall that we do not refer to default costs in the statements for results holding for $\alpha=\beta=1$.

Our first result exploits the fact that for edge-removal games without cash injections, the strategy profile where all edges are preserved is a (not necessarily unique) Nash equilibrium.
\begin{theorem} \label{always-existence}
Edge-removal games without cash injections always admit Nash equilibria. 
\end{theorem}
\begin{proof}
We claim that it is a Nash equilibrium if no bank removes an incoming edge. Indeed, consider a financial network and an arbitrary debt relation in that network represented by an edge $e$ from $v_i$ to $v_j$. Since there are no default costs, removing $e$ and keeping everything else unchanged, will result in $v_i$ instantly having $p_{ij}$ additional assets to pay to its lenders (excluding $v_j$). This, however can lead to at most $p_{ij}$ additional assets to reach $v_j$ through indirect paths starting from $v_i$ so it can not lead to more total payments from $v_i$ to $v_j$.\footnote{This would not be true if additional money was inserted to the network in the form of a cash injection.} Hence, $v_j$ is better off not removing $e$. Since $v_j$ is an arbitrary bank and $e$ represents an arbitrary edge, our proof is complete.
\end{proof}

\begin{theorem}\label{thm:eoa-eos}
In edge-removal games without cash injections, the negative Effect of Anarchy is unbounded and the negative Effect of Stability is $1$.
\end{theorem}
\begin{proof}

The result on the Effect of Stability follows by the proof of Theorem \ref{always-existence}, since the original network with no debt-removal is proved to be a Nash equilibrium. For the Effect of Anarchy, consider a simple network with two banks, each having a liability of $1$ to the other. In the original network, the total liquidity is $2$, while it is not hard to see that the network with both edges removed is a Nash equilibrium with a total liquidity of zero.
\end{proof}

Our next result shows that Nash equilibria may not exist once we allow for cash injections.
\begin{theorem}\label{thm:no-ne-dc}
There is an edge-removal game with cash injections that does not admit Nash equilibria.
\end{theorem}
\begin{proof}
Consider the network shown in Figure \ref{without_NE}, and a budget equals to $M=2-3\epsilon$, where $\epsilon$ is an arbitrarily small constant.
	
		\begin{figure}[htbp]
			\centering
			\includegraphics[scale=0.7]{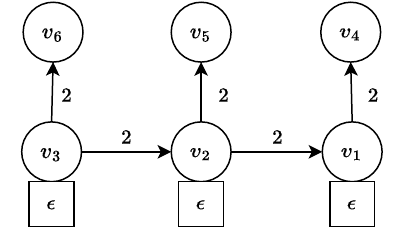}
			\caption{An edge-removal game without Nash equilibria in the case without default costs and with budget $M=2-3\epsilon$, where $\epsilon$ is an arbitrarily small constant.}
			\label{without_NE}
		\end{figure}

We begin by observing that $v_4, v_5$ and $v_6$ will never remove their incoming edges, since each of these agents has a single incoming edge, originating from some agent with positive externals. Hence, it suffices to consider the strategic actions of banks $v_1$ and $v_2$ regarding the possible removal of edges $(v_2,v_1)$ and $(v_3,v_2)$, respectively. There are four possible cases: 
\begin{itemize}
\item[A)] \emph{both edges are present.} In this case, $v_3$ has the highest threat index, i.e., $7/4$ and, thus receives the entire budget; it is not hard to verify that this is the optimal policy as well. The payments are, then, $p_{32}=1-\epsilon, p_{21}=1/2$, resulting in total assets $a_1=1/2+\epsilon, a_2=1$. 
\item[B)] \emph{$(v_2, v_1)$ is present, $(v_3, v_2)$ is removed.}  In this case, $v_2$ receives the budget and the following payments are realized $p_{32}=0, p_{21}=1-\epsilon$. This leads to total assets $a_1=1$ and $a_2=2-2\epsilon$. 
\item[C)] \emph{both edges are removed.} In this case, there is a tie on the maximum threat index ($v_1,v_2$ and $v_3$ have threat index $1$), so, assuming that the banks with lower index are prioritized, $v_1$ receives all the budget; again, it is not hard to verify that this is an optimal policy as well. This results to payments $p_{32}=0$ and $p_{21}=0$ with $a_1=2-2\epsilon, a_2=\epsilon$.
\item[D)] \emph{$(v_2, v_1)$ is removed, $(v_3, v_2)$ is present.} In this case, $v_3$ receives the budget. The payments are $p_{32}=1-\epsilon$ and $p_{21}=0$ (due to the removal), which implies $a_1=\epsilon$ and $a_2=1$ respectively. 
\end{itemize}

One can now easily check that the best response dynamics cycle, as starting from Case A, $v_2$ has an incentive to remove its incoming edge $(v_3, v_2)$ and we reach Case B. Then, $v_1$  has an incentive to remove its incoming edge $(v_2, v_1)$ (we are now in Case C), which leads $v_2$ to have an incentive to reinstate its incoming edge $(v_3, v_2)$ (thus, reaching Case D). Finally, in Case D, $v_1$ has an incentive to reinstate $(v_2, v_1)$, leading to Case A again. 
\end{proof}

\begin{theorem}\label{thm:pos}
The Price of Stability in edge-removal games (with or without cash injections) is unbounded. 
\end{theorem}
\begin{proof}
We present the proof for the case without cash injections and note that the proof carries over for the case of limited budget regardless of who receives it. Moreover, note that the proof does not assume default costs but the result immediately applies to that (more general) case as well, i.e., $\alpha=\beta=1$ is a special case of default costs.

		\begin{figure}[htbp]
			\centering
			\includegraphics[scale=0.7]{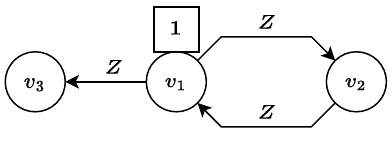}
			\caption{A financial network with no default costs that admits unbounded price of stability.}
			\label{Unbouned PoS without default cost}
		\end{figure}

Consider the network shown in Figure \ref{Unbouned PoS without default cost}, where $Z$ is an arbitrary large constant. Since each bank has exactly one incoming edge, no edge removals occur at the unique Nash equilibrium. By the proportionality principle it holds that the three payments have to be equal (to $1$), that is $p_{13}=p_{12}=p_{21}=1$,  which results in the systemic liquidity of $3$. However, systemic liquidity of $2Z$ can be achieved when $v_3$ removes its incoming edge. Indeed, then each remaining payment can be equal to $Z$. We conclude that $\pos \geq 2Z/3$, so the price of stability can be arbitrarily large for appropriately large values of $Z$, as desired.
\end{proof}

\begin{theorem}\label{eoa-order-of-n}
The negative Effect of Stability in edge-removal games with cash injections is $\Omega(n)$.
\end{theorem}

\begin{proof}

Consider the network in Figure \ref{PoS} where the budget $M=1$, $k=n/2$ and $Z$ is arbitrarily larger than $k$.
		\begin{figure}[htbp]
		\centering
		\includegraphics[scale=0.7]{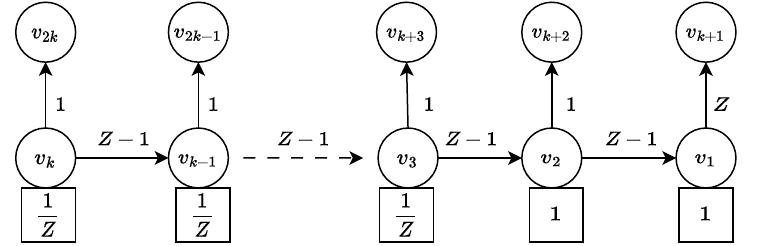}
		\caption{A network that yields the negative Effect of Stability of $\Omega(n)$ for budget $M=1$, $n=2k$, and arbitrarily large $Z$.}
		\label{PoS}
	\end{figure}
We start by noticing that $\mu_1=1$, while for $i=2, \dots, k$, it holds that  $\mu_{i}=1+\frac{Z-1}{Z}\mu_{i-1}\approx 1+\mu_{i-1}$, for sufficiently large $Z$; all other banks are solvent. Hence, the total liquidity of the original network, which is in fact the optimal one, is achieved when $v_k$ receives the entire budget of $M=1$ as a cash injection, and 
is roughly $kM=n/2$, when $Z$ is sufficiently large. 

We now claim that under any Nash equilibrium, $v_2$ will receive the  budget and all edges $(v_i,v_{i-1})$ for $i\in \{3, \dots, k\}$ are removed. This would complete the proof, as the total liquidity would be at most $\frac{1}{Z}(k-2)+2+2\leq 5$. We now prove this claim. Consider any equilibrium and observe that $v_{i}$, for $i=k+1, \dots,  2k$, must have their unique incoming edge present. Now, assume for a contradiction that some bank $v_i$ with $i \in \{3, \dots, k\}$ gets a cash injection; this implies that the edge $(v_i,v_{i-1})$ is present as, otherwise, the result holds trivially. Then, bank $v_{i-1}$  has total assets $1-1/Z^2+e_{i-1}$, but can increase them to $1+e_{i-1}$ by strategically removing its incoming edge. So, under any Nash equilibrium, either $v_2$ or $v_1$ receives a cash injection. In the former case, where edge $(v_2,v_1)$ is present, $a_1=3-2/Z$, while the assets of $v_1$ would be $2$ if it removed its incoming edge and received the cash injection. 

So far, we have proven that $v_2$ gets the cash injection and it remains to show no other edge $(v_i, v_{i-1})$ for $i\in \{3,\dots, k\}$ exists in a Nash equilibrium. Now, observe that if such an edge exists, then neighboring edges on the horizontal path cannot exist as that would contradict that $v_2$ gets the cash injection. Then, when $i>4$, bank $v_{i-2}$ would have an incentive to add edge $(v_{i-1}, v_{i-2})$, thus, making bank $v_{i}$ the recipient of the budget (for both optimal and greedy) and strictly increase its own total assets. The cases $i\in \{3,4\}$ can be easily ruled out as well.
Our proof is complete.
\end{proof}

We now present a series of results for the case where default costs exist, but cash injections are not allowed.
Contrary to the case with neither default costs nor cash injections, we show that a Nash equilibrium is no longer guaranteed to exist; the next result is complementary to Theorem~\ref{thm:no-ne-dc}.

\begin{theorem}\label{thm:gam-7}
There is an edge-removal game without cash injections but with default costs that does not admit Nash equilibria.
\end{theorem}
\begin{proof}
We present the proof for the case without cash injections and note that the proof carries over for the case of limited budget regardless of who receives it, as a sufficiently small budget will not alter the players incentives.

\begin{figure}[ht]
		\centering
		\includegraphics[scale=0.8]{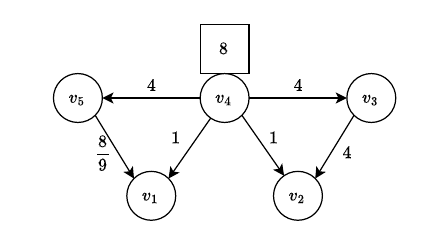}
		\caption{An instance of an edge-removal game without Nash equilibria, when $\alpha =\beta =1/4$.}
		\label{without neither bailout nor NE}
	\end{figure}
	
 Consider the financial network in Figure \ref{without neither bailout nor NE} and assume that default costs $\alpha =\beta =1/4$ are applied. 	
We begin by claiming that $v_5$ (and, by symmetry $v_3$) will never remove its incoming edge. Indeed, since $v_4$ has positive external assets, then $v_5$ will receive a positive payment from $v_4$ if the edge between them remains, however $v_5$ will have zero incoming payments if the corresponding edge is removed. Similarly, $v_1$ (and, by symmetry $v_2$) will never remove the edge from $v_5$ (respectively $v_3$). Therefore, it suffices to consider the possible removal of edges $(v_4,v_1)$ and $(v_4, v_2)$. We will prove that none of the following states are at equilibrium:  A) no edge is removed, B) $(v_4, v_1)$ is removed but $(v_4, v_2)$ remains, C) both edges are removed, and D) $(v_4, v_1)$ remains but $(v_4, v_2)$ is removed. 

Note that in Case A $v_4$ is in default, hence, due to the default costs, its payments are broken down as $\mathbf{p}_4=(1/5,1/5,4/5,0,4/5)$. But then, $v_3$ and $v_5$ are also in default, and $p_{51}=p_{32}=1/5$ in this case. So, the utility of both $v_1$ and $v_2$ is $2/5$. By removing $(v_4, v_1)$ in Case A, and moving to Case B, $v_1$ would increase its utility to $8/9$. Indeed, in this case, $\mathbf{p}_4=(0,2/9,8/9,0,8/9)$, while $p_{51}=8/9$ and $p_{32}=2/9$, so the utility of $v_1$ is $8/9$ and the utility of $v_2$ is $4/9$. But, then, it is beneficial for $v_2$ to remove $(v_4, v_2$).  In this case (Case C), $v_4$ is now solvent and all existing debts are paid for, thus giving utility $8/9$ to $v_1$ and utility $4$ to $v_2$. However, if $v_1$ then decides to reinstate its incoming edge $(v_4, v_1)$, its utility in this case (Case D) is increased to $10/9$, while $v_2$'s utility is $2/9$, since  $\mathbf{p}_4=(2/9,0,8/9,0,8/9)$, $p_{51}=8/9$ and $p_{32}=2/9$. Finally, $v_2$ prefers its utility in Case A to its utility in Case D so will decide to reinstate $(v_4, v_2)$ when in Case D, thus defining a cycle between the four possible states.
\end{proof}

For some restricted topologies, however, the existence of Nash equilibria is guaranteed; in particular, keeping all edges is a Nash equilibrium. 
\begin{theorem}\label{thm:dag-ne} 
Edge-removal games without cash injections but with default costs always admit Nash equilibria if the financial network is a tree or a cycle.
\end{theorem}
\begin{proof}
We present the proof for the case without cash injections and note that the proof carries over for the case of limited budget regardless of who receives it.

Consider a financial network that is a tree (similar argument works in case of a directed cycle) and an arbitrary debt relation in that network represented by a directed edge $(v_i, v_j)$. By definition of the tree structure, it holds that there is no other path in the network  from $v_i$ to $v_j$. Hence, $v_j$ cannot benefit by removing that edge, even if the removal affects (increases) $v_i$'s available assets.
\end{proof}
The following result demonstrates that the positive impact of (individually benefiting) edge removals dominates the negative impact of reducing the number of edges through which money can flow, hence, edge removals are in line with the regulator's best interest too.
\begin{lemma}\label{lem:impact_of_removal}
Edge-removal games with default costs but no cash injections satisfy the following: given any network and any strategy profile, any unilateral removal of any edge(s) that weakly improves the total assets of the corresponding bank, also weakly improves the total assets of every other bank in the network. Consequently, the total liquidity of the system is increased.
\end{lemma}
\begin{proof}
Consider a network $\fn=( V, E)$ and a strategy profile $\mathbf{s}=(s_1, \ldots,s_n)$, under which banks have total assets according to $\mathbf{a}=(a_1,\ldots, a_n)$. Fix a bank $i$ and let $\mathbf{s}'=(s_i', \mathbf{s}_{-i})$ be the strategy profile that is derived by $\mathbf{s}$ if bank $i$ changes its strategy from $s_i$ to $s_i'$, where $s_i'$ is derived by $s_i$ by the removal of an edge $e=(v_j,v_i)$ (the argument can be applied repeatedly to prove the claim for more than one edge removals). By assumption, the total assets of bank $i$ under $\mathbf{s}'$, $a_i'$, satisfy $a_i'\geq a_i$. It holds that any bank reachable by $i$ or $j$ (the two endpoints of the edge that was removed) through a directed path will have at least the same total assets under $\mathbf{s}'$ than with $\mathbf{s}$, since there will be at least the same amount of money available to leave $i$ and $j$ and traverse these paths. The assets of banks not reachable by $i$ or $j$ will, clearly, not be affected by the removal of $e$. Hence, the assets of each bank in $\fn$ are weakly higher under $\mathbf{s}'$ than under $\mathbf{s}$.
The increase in the total liquidity follows since the total assets, by definition,  equal external assets plus payments. 
\end{proof}

In fact, the systemic liquidity of even the worst Nash equilibrium can be arbitrarily higher than at the original network. To see this, consider the network in the proof of Theorem \ref{thm:eos-unbounded}, which admits a unique Nash equilibrium with arbitrarily higher total liquidity than that of the original network.
\begin{theorem}\label{thm:eos-unbounded}
The positive Effect of Anarchy in edge-removal games with default costs and without cash injections is unbounded.
\end{theorem}
\begin{proof}

Consider the network in Figure \ref{tiny FoA} where default costs are $\alpha=\beta=\epsilon$ for some arbitrary small positive number $\epsilon$. If no edge is removed, then all banks except $v_n$ are in default and the following payments are realized: $p_{12}=p_{1n}=\epsilon/2$ and $p_{i,i+1}=\epsilon\cdot p_{i-1,i}=\epsilon^{i}/2.$ The systemic liquidity is then $\liq_{\fn}=\epsilon/2+\sum_{i=1}^{n-1}\frac{\epsilon^i}{2}<\frac{\epsilon}{1-\epsilon}$.
	\begin{figure}[ht]
		\centering
		\includegraphics[scale=0.7]{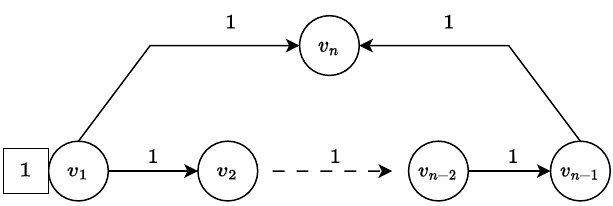}
		\caption{A financial network with $\alpha =\beta =\epsilon$, for an arbitrarily small positive number $\epsilon$, that admits a unbounded positive Effect of Anarchy.}
		\label{tiny FoA}
	\end{figure}
On the other hand, the unique Nash equilibrium is achieved when $v_n$ removes the edge pointing from $v_1$ to itself. The systemic liquidity in this case is $n-1$, and the proof follows.
\end{proof}

We conclude with our results on computational complexity for the setting with default costs.
\begin{theorem} \label{NP-hard in network with default}
In edge-removal games with default costs, the following problems are \emph{\NP}-hard: 
\begin{enumerate}[label=\emph{\alph*})]
\item Decide whether a Nash equilibrium exists.
\item Verify if a given strategy profile is a pure Nash equilibrium when pure equilibria are guaranteed to exist.
\item Compute a best-response strategy.
\item Compute a strategy profile that maximizes systemic liquidity. 
\end{enumerate}
\end{theorem}
\begin{proof}
We begin by proving that the problem of verifying whether a strategy profile is a Nash equilibrium (even when its existence is guaranteed) is \NP-hard. 
\medskip

\textbf{Hardness of verifying a Nash equilibrium.} Our proof follows by a reduction from the \textsc{Partition} problem. Recall that in \textsc{Partition}, an instance $\mathcal{I}$ consists of a set $X$ of positive integers $\{x_1, x_2, \dots, x_k\}$ and the question is whether there exists a subset $X'$ of $X$ such that $\sum_{i \in X'}{x_i} = \sum_{i \notin X'}{x_i} = \frac{\sum_{i \in X}{x_i}}{2}$.

The reduction works as follows. Starting from $\mathcal{I}$, we build an instance $\mathcal{I}'$ by adding an agent $v_i$ for each element $x_i \in X$ and allocating an external asset of $e_i = x_i$ to $v_i$; we also include two additional agents $S$ and $T$. Each agent $v_i$ has liability equal to $e_i$ to each of $S$ and $T$, while $T$ has liability $\frac{\sum_{i}{e_i}}{2}+\frac{1}{4}$ to $S$; see also Figure \ref{fig:compute_hard}. We furthermore set default costs $\alpha=\beta = 1/\sum_i{e_i}$; clearly, the reduction requires polynomial-time. 

Observe that, in any Nash equilibrium, agent $T$ keeps all its incoming edges. Indeed, removing an edge from agent $v_i$ will decrease $T$'s total assets, as there is no alternative path for payments originating at $v_i$ to reach $T$. Similarly, agent $S$ keeps its incoming edge from $T$ at any Nash equilibrium, as deleting it will reduce $S$'s total assets. Therefore, the only strategic choice in this financial networks is by agent $S$ about which edges  from agents $v_i$ to keep and which to remove. We denote by $S_k$ and $S_r$ the set of agents whose edges towards $S$ are kept and removed, respectively, and observe that agents in $S_k$ are in default while agents in $S_r$ are not. Clearly, as $S$ is essentially the only strategic agent, any best-response strategy by $S$ forms a Nash equilibrium. This guarantees the existence of a Nash equilibrium in $\mathcal{I}'$. We will show that in instance $\mathcal{I}'$, agent $S$ can compute her best response, and hence we can compute a Nash equilibrium, if and only if instance $\mathcal{I}$ of \textsc{Partition} is a yes-instance.

We first show that if $\mathcal{I}$ is a yes-instance for \textsc{Partition}, then agent $S$ has total assets $a_S = \frac{\sum_i{e_i}}{2}+\frac{1}{2}$. Indeed, consider the subset $X'$ in $\mathcal{I}$ with $\sum_{i \in X'}{x_i} = \frac{\sum_{i \in X}{x_i}}{2}$ and let $S_k$ contain agents $v_i$ where $x_i \in X'$, while $S_r$ agents $v_i$ with $x_i \notin X'$. Note that $T$ obtains total assets 
\begin{align*}
a_T &= \sum_{i \in S_r}{e_i}+\frac{1}{\sum_i{e_i}}\sum_{i \in S_k}{\frac{e_i}{2}}\\
&= \frac{\sum_i{e_i}}{2}+\frac{1}{4},
\end{align*}
and is therefore solvent, while $S$ obtains 
\begin{align*}
a_S &= \frac{1}{\sum_i{e_i}}\sum_{i \in S_k}{\frac{e_i}{2}} + \frac{\sum_i{e_i}}{2}+\frac{1}{4}\\
&= \frac{\sum_i{e_i}}{2}+\frac{1}{2}.
\end{align*}

We will now show that if $\mathcal{I}$ is a no-instance, then $S$'s total assets in $\mathcal{I}'$ are strictly less than $\frac{\sum_i{e_i}}{2}+\frac{1}{2}$; this suffices to prove the claim. Consider any subset $X' \subseteq X$ in $\mathcal{I}$ and the corresponding strategy profile in $\mathcal{I}'$ where $S$ keeps incoming edges from agents in $S_k$ while removes edges from agents in $S_r$. Let $\chi = \sum_{i \in S_r}{e_i}$ and observe that, as elements in $X$ are integers, it holds either $\chi \leq \frac{\sum_i{e_i}}{2}-\frac{1}{2}$ or $\chi \geq \frac{\sum_i{e_i}}{2}+\frac{1}{2}$. 

In the first case, when $\chi \leq \frac{\sum_i{e_i}}{2}-\frac{1}{2}$, we claim that $T$ is in default. Indeed, $T$ collects a total payment of $\chi$ from the agents in $S_r$ and a total payment of $\frac{1}{\sum_i{e_i}}\sum_{i \in S_k}{\frac{e_i}{2}} = \frac{1}{\sum_i{e_i}}\frac{\sum_i{e_i - \chi}}{2} = \frac{1}{2}-\frac{\chi}{2\sum_i{e_i}}$ from the agents in $S_k$. So, 
\begin{align*}
a_T &= \chi + \frac{1}{2} -\frac{\chi}{2\sum_i{e_i}}\\
&= \chi(1-\frac{1}{2\sum_i{e_i}}) +\frac{1}{2}\\
&\leq \left(\frac{\sum_i{e_i}}{2}-\frac{1}{2}\right) \left(1-\frac{1}{2\sum_i{e_i}}\right) + \frac{1}{2}\\
&< \frac{\sum_i{e_i}}{2},
\end{align*}
i.e., less than $T$'s liability to $S$; the first inequality follows by the assumption on $\chi$.
 Therefore, $S$ can collect $\frac{1}{\sum_i{e_i}}\frac{\sum_i{e_i-\chi}}{2} $ from the agents in $S_k$ and strictly less than $\frac{1}{2}$ from $T$. We conclude that $a_S < \frac{\sum_i{e_i}}{2}+\frac{1}{2}$ in this case. 

In the second case, when $\chi \geq \frac{\sum_i{e_i}}{2}+\frac{1}{2}$, $S$ obtains a total payment of $\frac{1}{\sum_i{e_i}}\sum_{i \in S_k}{\frac{e_i}{2}} = \frac{1}{2\sum_i{e_i}}(\sum_i{e_i}-\chi) \leq \frac{1}{2\sum_i{e_i}}(\frac{\sum_i{e_i}}{2}-\frac{1}{2}) = \frac{1}{4}-\frac{1}{4\sum_i{e_i}}$ from agents in $S_k$ and a payment of at most $\frac{\sum_i{e_i}}{2}+\frac{1}{4}$ from $T$, i.e., $a_S < \frac{\sum_i{e_i}}{2}+\frac{1}{2}$ again.

Since in any case, a no-instance $\mathcal{I}$ for \textsc{Partition} leads to an instance $\mathcal{I'}$ where $a_S < \frac{\sum_i{e_i}}{2}+\frac{1}{2}$, the claim follows, as we cannot compute a best-response strategy for $S$, and, by the discussion above, a Nash equilibrium.
\medskip

\textbf{Hardness of computing a best-response strategy.} The proof was given in the previous case.

\begin{figure}[ht]
	\centering
	\includegraphics[scale=0.7]{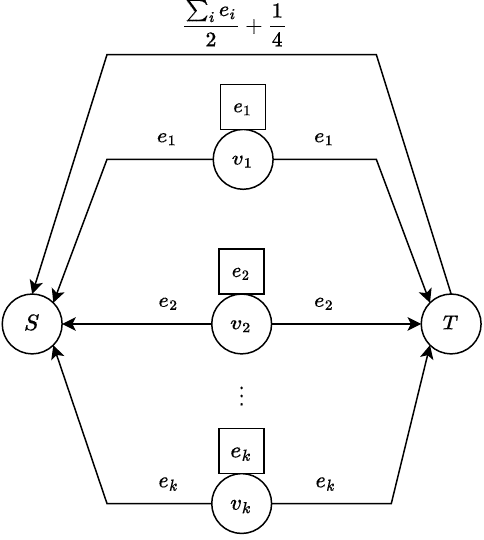}
	\caption{The instance arising from the reduction used in Theorem \ref{NP-hard in network with default}b.}
	\label{fig:compute_hard}
\end{figure}
\medskip

\textbf{Hardness of maximizing systemic liquidity.} The proof follows by the reduction from \textsc{Partition} described in the previous case by adding a path of $\ell$ agents as shown in Figure \ref{fig:compute_opt_hard}. By the discussion above, starting from a yes-instance in \textsc{Partition}, we have $a_S=\frac{ \sum_i e_i }{2}+\frac{1}{2}$, and $S$ as well as any agent $u_i$ are solvent. On the contrary, starting from a no-instance for \textsc{Partition}, agent $S$ is in default and the payments traveling to the $u_i$ agents get reduced by a factor of $\sum_i{e_i}$ at each edge. By selecting $\ell$ to be large enough, the claim follows.
\begin{figure}[ht]
	\centering
	\includegraphics[scale=0.7]{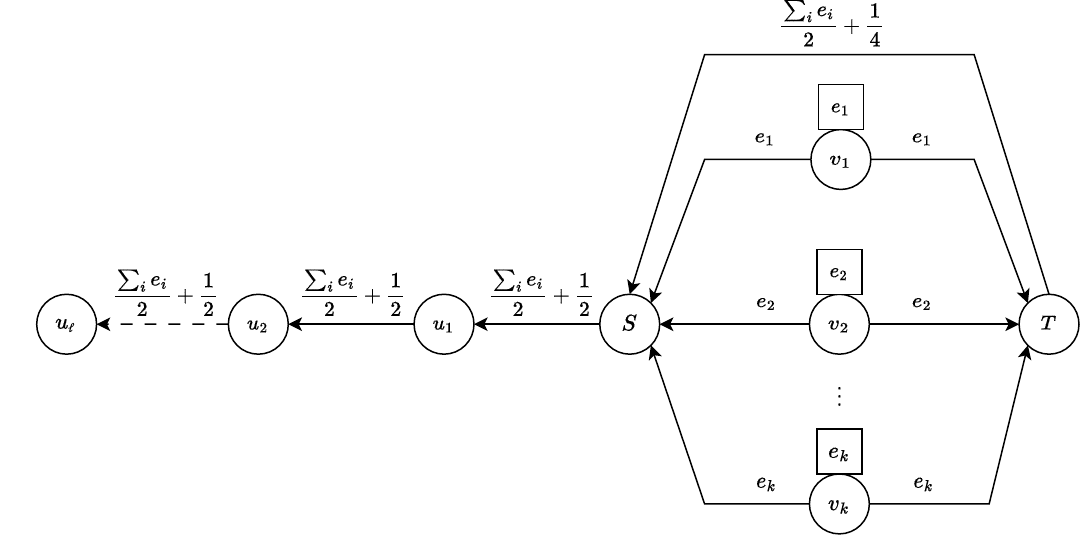}
	\caption{The modified instance in the proof of Theorem \ref{NP-hard in network with default}d.}
	\label{fig:compute_opt_hard}
\end{figure}
\medskip

\textbf{Hardness of deciding the existence of Nash equilibria.} 
Again, the proof follows by the reduction from \textsc{Partition} used in Theorem \ref{NP-hard in network with default}b by adding five agents, $\ell, m, r, x$ and $y$ with liabilities and external assets are show in Figure~\ref{fig:hard-existence}. Recall that the default costs are $\alpha=\beta = \frac{1}{\sum_i e_i}$ and, without loss of generality, we assume that $\sum_i e_i\geq3$.

\begin{figure}[ht]
	\centering
	\includegraphics[scale=0.7]{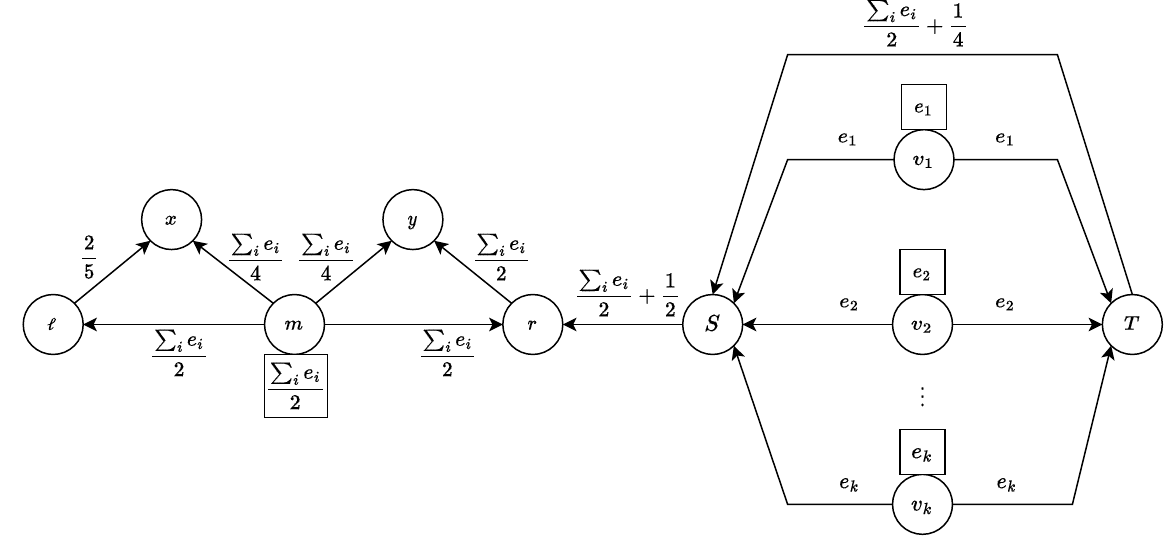}
	\caption{The instance used in the proof of Theorem \ref{NP-hard in network with default}a.}
	\label{fig:hard-existence}
\end{figure}

As argued earlier, $T$ always keeps all its incoming edges, while $S$ always keeps the edge from $T$. Similarly, agents $\ell$ and $r$ always keep their incoming edges, as by removing any incoming edge their total assets strictly decrease. As agents $v_i$ and $m$ do not have incoming edges, the only strategic agents are $x, y$ (with respect to the edges originating from $m$) and $S$ (with respect to edges from the $v_i$ agents). We first show that, if instance $\mathcal{I}$ of \textsc{Partition} is a yes-instance, then there is a Nash equilibrium. Indeed, we argued earlier that in this case $a_S =\frac{\sum_i{e_i}}{2}+\frac{1}{2}$ and, hence, $S$ is solvent. So, $r$ is solvent as well and, therefore, $y$ always keeps the edge from $m$ as the liability from $r$ is always fully paid. This implies that $m$ is necessarily in default and agent $x$ should keep the edge from $m$ as well. To conclude, the strategy profile where $x, y$ keep all edges is a Nash equilibrium. 

On the other hand, when instance $\mathcal{I}$ is a no-instance, we have shown that $S$ is in default. Then, $r$ collects a payment of $\chi \leq \frac{1}{2}+\frac{1}{2\sum_i{e_i}}$ from $S$ and is, hence, necessarily in default. When both $x$ and $y$ keep their edges from $m$, their total assets are $a_x = \frac{1}{3\sum_i{e_i}}+\frac{1}{6}$ and $a_y = \frac{1}{6}+(\frac{1}{3}+\chi)\frac{1}{\sum_i{e_i}}$. When $x$ removes the edge and $y$ keeps it, it is $a_x = \frac{2}{5}$ and $a_y = \frac{1}{5}+(\frac{2}{5}+\chi)\frac{1}{\sum_i{e_i}}$, i.e., $x$ improves compared to the previous case.
When they both remove these edges, we have $a_x = \frac{2}{5}$ and $a_y = \frac{\sum_i{e_i}}{2}$, i.e., $y$ improves compared to the previous case. Similarly,  when $x$ keeps the edge and $y$ removes it, the total assets are $a_x = \frac{3}{5}$ and $a_y = (\frac{2}{5}+\chi)\frac{1}{\sum_i{e_i}}$, i.e., $x$ improves compared to the previous case. The claim follows by observing that, in the last case, $y$ improves by keeping the edge. 
\end{proof}
\section{Conclusions}\label{sec:conclusions}
We considered problems arising in financial networks, when a financial authority wishes to maximize the total liquidity either by injecting cash or by removing debt. We also studied the setting where banks are rational strategic agents that might prefer to forgive some debt if this leads to greater utility, and we analyzed the corresponding games with respect to properties of Nash equilibria. In that context, we also introduced the notion of the Effect of Anarchy (Stability, respectively) that compares the liquidity in the initial network to that of the worst (best, respectively) Nash equilibrium,while  distinguishing between the cases where the equilibria (stable network structures) are better or worse than the original network.

Our work leaves some interesting problems unresolved. Given the computational hardness of some of the optimization problems, it makes sense to consider approximation algorithms. From the game-theoretic point of view, one can also consider the problems from a mechanism design angle, i.e., to design incentive-compatible policies where banks weakly prefer to keep all incoming liabilities. It would also be interesting to identify classes of games where all equilibria are better (or all equilibria are worse, respectively) than the original network, thus demonstrating only a positive (or only a negative, respectively) Effect of Anarchy/Stability.

\section*{Acknowledgements}
\textbf{Hao Zhou} acknowledges funding from Trustworthy AI - Integrating Learning, Optimisation and Reasoning (TAILOR), a project funded by European Union Horizon2020 research and innovation program under Grant Agreement 952215.

	





\bibliographystyle{plainnat}
\bibliography{ref}


\end{document}